\documentclass[%
 reprint,
 showkeys,
 amsmath,amssymb,
 aps,
pra,
]{revtex4-1}

\usepackage{graphicx}
\usepackage{dcolumn}
\usepackage{bm}
\usepackage{braket}
\usepackage{amsthm}
\usepackage{multirow}
\usepackage{algorithm,algorithmic}
\usepackage{mathtools}
\usepackage{enumerate}

\theoremstyle{plain}
\newtheorem{thm}{Theorem}
\newtheorem{lemma}{Lemma}
\theoremstyle{definition}

\begin{document}


\title{Exploiting degeneracy to construct good ternary quantum error correcting code}

\author{Ritajit Majumdar} 
\email{majumdar.ritajit@gmail.com}
\author{Susmita Sur-Kolay}
\email{ssk@isical.ac.in}

\affiliation{Advanced Computing \& Microelectronics Unit, Indian Statistical Institute, India}


\begin{abstract}
Quantum error-correcting code for higher dimensional systems can, in general, be directly constructed from the codes for qubit systems. What remains unknown is whether there exist efficient code design techniques for higher dimensional systems. In this paper, we propose a 7-qutrit error-correcting code for the ternary quantum system and show that this design formulation has no equivalence in qubit systems. This code is optimum in the number of qutrits required to correct a single error while maintaining the CSS structure. This degenerate CSS code can (i) correct up to seven simultaneous phase errors and a single bit error, (ii) correct two simultaneous bit errors on pre-defined pairs of qutrits on eighteen out of twenty-one possible pairs, and (iii) in terms of the cost of implementation, the depth of the circuit of this code is only two more than that of the ternary Steane code. Our proposed code shows that it is possible to design better codes explicitly for ternary quantum systems instead of simply carrying over codes from binary quantum systems.
\end{abstract}

\maketitle


\section{\label{sec:intro}Introduction}

The design of a large-scale, general-purpose quantum computer is hindered by error. The interaction of qubits with the environment readily destroys the information content of the state. As a remedy, quantum error correction techniques have been proposed \cite{PhysRevA.52.R2493, PhysRevLett.77.793, PhysRevLett.77.198, gottesman1997stabilizer, kitaev2003fault} which encodes the information of $k$ qubits into $n ~(> k)$ qubits to correct errors. However, the fidelity of a quantum circuit decays exponentially with the depth of the circuit \cite{arute2019quantum}. Therefore, quantum error-correcting code (QECC) with higher circuit depth not only reduces the speed of computation but also has the potential to incorporate further errors. The requirement for QECC construction is, therefore, to reduce the number of qubits required for encoding, as well as reduce the depth of the QECC circuit.

Consider two linear classical codes $C_1 = [n,k_1,d_1]$ and $C_2 = [n,k_2,d_2]$ such that $C_2^{\perp} \subseteq C_1$ and $k_2 < k_1$. The parity check matrices of $C_1$ and $C_2$ can be combined to construct an $[[n,k_1 - k_2, min\{d_1,d_2\}]]$ QECC \cite{PhysRevA.54.1098,PhysRevLett.77.793}, called the CSS code. The parity check matrix of $C_1 (C_2)$ forms the stabilizers for bit error correction, while that of $C_2 (C_1)$ forms the stabilizers for phase error correction. Therefore, the stabilizers of a CSS code can be partitioned into two sets such that the non-identity operators in the two partitions are different. CSS codes usually have low depth circuit as compared to non-CSS codes \cite{devitt2013quantum,majumdar2017method}. For the rest of this paper, we shall concentrate on CSS codes only.

Quantum systems are inherently multi-valued, and multi-valued quantum computers can outperform their binary counterparts in certain cryptographic protocols \cite{PhysRevLett.85.3313}, search algorithms \cite{wong2015grover}, and machine learning \cite{adhikary2020supervised}. QECC for multi-valued systems can be directly carried over from the QECC construction for qubit systems \cite{PhysRevA.55.R839, PhysRevA.97.052302}. To the best of our knowledge, it is not known whether it is possible to explicitly construct more efficient QECC for a particular higher-dimensional quantum system, instead of using the construction technique from some lower-dimensional system.

A ternary quantum system (or qutrit) is the simplest multi-valued system. The 7-qubit Steane code \cite{PhysRevLett.77.793} can be carried over to the ternary system, which is the optimal CSS QECC in the number of qutrits that corrects a single error \cite{majumdar2020approximate}. In this paper we propose a 7-qutrit degenerate CSS QECC which is (i) optimal in the number of qutrits (for CSS type code), (ii) can correct up to seven simultaneous phase errors, (iii) can correct a single bit error, (iv) can be readily designed to correct two simultaneous bit errors on a predefined pair of qutrits (on 18 out of 21 pairs), and (v) the depth of the circuit of our QECC is only two more than that of the ternary Steane code. Its ability to correct multiple phase errors and two simultaneous bit errors on pre-defined qutrit pairs without a significant increase in the depth of the circuit makes our proposed QECC a strong candidate for error correction in the ternary quantum system. We also show that our proposed construction has no equivalence in the qubit system in order to construct a linear degenerate 7-qubit QECC. Therefore, our QECC implies that there exist design techniques of QECC for ternary quantum systems which are more efficient than ternary QECCs which are carried over from binary QECCs.

The rest of the paper is organized as follows - Section II gives a brief description of quantum error correction using stabilizers followed by the construction of the circuit of ternary Steane code in Section III. In Section IV we introduce our QECC and describe phase and bit error correction in detail. Section V compares the circuit cost of our QECC to the ternary Steane code. We conclude in Section VI.

\section{Preliminaries of quantum error correction}

\subsection{Error model and degeneracy}
The error model considered in this paper
\begin{equation} \label{eq:model}
\mathcal{E} = \delta\mathbb{I}_3 + \sum_{i = 1}^2\eta_iZ_i + \sum_{j=1}^{2}(\mu_jX_j + \sum_{i,j}\xi_{ij}Y_{ij})
\end{equation}
where $\mathbb{I}_3$ is the $(3 \times 3)$ identity matrix,
\begin{eqnarray*}
    Z_i\ket{\psi} &=& \alpha\ket{0} + \omega^i\beta\ket{1} + \omega^{2i}\gamma\ket{2}\\
    X_i\ket{\psi} &=& \alpha\ket{0+i} + \beta\ket{1+i} + \gamma\ket{2+i}\\
     Y_{ij} &= &X_iZ_j
\end{eqnarray*}
for $i,j \in \{1,2\}$, $\delta,\eta,\mu,\xi \in \mathbb{C}$, spans the $(3 \times 3)$ operator space \cite{majumdar2019near}. $X_i$ and $Z_i$ are termed as bit and phase errors respectively.

For binary quantum systems, a set of mutually commuting operators $S_1, \hdots, S_m \subset$ \{$I$, $\pm \sigma_x$, $\pm i\sigma_x$, $\pm \sigma_y$, $\pm i\sigma_y$, $\pm \sigma_z$, $\pm i\sigma_z$\}$^{\otimes n}$, where each $\sigma_i$ is a Pauli operator \cite{nielsen2002quantum}, is said to stabilize a quantum state $\ket{\psi}$ if \cite{gottesman1997stabilizer}:

\begin{enumerate}
        \item ${\forall i}$, $S_i\ket{\psi} = \ket{\psi}$, $1\le i \le m$;
        \item ${\forall e} \in \mathcal{E}$, $\exists$ $j$, $S_j(e\ket{\psi}) = -(e\ket{\psi})$  $1 \leq j \leq m$;
        \item for $e,e' \in \mathcal{E}$, $e \neq e'$, $\exists$ $j,k$ $S_j(e\ket{\psi}) \neq S_k(e'\ket{\psi})$, $1 \leq j, k \leq m$;
\end{enumerate}

Weight of a stabilizer $S_i$ ($wt(S_i)$) is defined as the number of non-identity operators in $S_i$. For an $[[n,k,d]]$ QECC, let S be the set of stabilizers and $\mathcal{O}$ be the set of all n-qubit operators. The distance $d$ of the QECC is defined as
\begin{center}
    $min \{wt(i)$ $|$ $i \in \mathcal{O} \setminus S$ and $[i,S_j] = 0$ $\forall$ $S_j \in S\}$
\end{center}

A distance $d$ QECC can correct upto $\lfloor\frac{d}{2}\rfloor$ errors. A QECC is said to be degenerate if $\exists$ $e, e' \in \mathcal{E}$, $e \neq e'$, where $\mathcal{E}$ is the set of all correctable errors, such that for a codeword $\ket{\psi}$, $e\ket{\psi} = e'\ket{\psi} = \ket{\phi}$. For such a code, it is not necessary to distinguish between $e$ and $e'$ as long as the error state $\ket{\phi}$ can be uniquely identified.

In accordance to the error model of Eq.~\ref{eq:model}, the operators which form the stabilizers for ternary quantum system are
\begin{center}
    $X_1\ket{j} = \ket{j+1}$ mod $3$ \quad $X_2\ket{j} = \ket{j+2}$ mod $3$\\
    $Z_1\ket{j} = \omega^j\ket{j}$ \quad $Z_2\ket{j} = \omega^{2j}\ket{j}$
\end{center}
where $j \in \{0,1,2\}$, and $\omega^3 = 1$. It can be noted that
\begin{center}
    $X_2 = X_1X_1$ \quad $Z_2 = Z_1Z_1$.
\end{center}

The ternary stabilizers for any $[[n,k,d]]_3$ QECC will be $n$-fold tensor products of $\{I,X_1,X_2,Z_1,Z_2\}$ \cite{majumdar2020approximate}.

\begin{lemma}
\label{eq:l1}
$[X_i \otimes X_j, Z_k \otimes Z_l]$ = 0 if and only if 
\begin{center}
    $\begin{cases}
    i = j~ \text{when}~ k \neq l,\\
    i \neq j~ \text{when}~ k = l.
\end{cases}$
\end{center}
\end{lemma}

\begin{proof}
It is easy to verify that 
\begin{equation}
\label{eq:commute}
    X_i Z_i = \omega^2 Z_i X_i, i \in \{1,2\}
\end{equation}

From Eq.~\ref{eq:commute}, the following can be derived
\begin{center}
    {\small $X_1 Z_2 = X_1 Z_1 Z_1 = \omega^2 Z_1 X_1 Z_1 = \omega^2 Z_1 (\omega^2 Z_1 X_1) = \omega Z_2 X_1$}\\
    {\small $X_2 Z_1 = X_1 X_1 Z_1 = X_1 (\omega^2 Z_1 X_1) = \omega^2 \omega^2 Z_1 X_1 X_1 = \omega Z_1 X_2$}
\end{center}

In summary,
\begin{equation}
\label{eq:commute2}
    X_i Z_j = \omega Z_j X_i, i,j \in \{1,2\}, i \neq j
\end{equation}

Now consider $(X_i \otimes X_j)(Z_k \otimes Z_l)$ if either $i=j, k \neq l$ or $i \neq j, k = l$, $i,j,k,l \in \{1,2\}$. Thus, either $i=k, j \neq l$ or $i \neq k, j=l$. From Eq.~\ref{eq:commute} and \ref{eq:commute2}, commutation on one of the qutrits yields $\omega^2$, and  on the other qutrit yresults in $\omega$, and hence the product is 1. Therefore, for such a scenario, $[X_i \otimes X_j, Z_k \otimes Z_l]$ = 0.

Conversely, consider $(X_i \otimes X_j)(Z_k \otimes Z_l)$ where $i=j$ and $k=l$. Then from Eq.~\ref{eq:commute} and \ref{eq:commute2}, if $i=k$ ($i \neq k$), then commutation on both the qutrits produce value $\omega^2$ ($\omega$), and hence the product is $\omega$ ($\omega^2$). Therefore, for such a scenario, $[X_i \otimes X_j, Z_k \otimes Z_l] \neq 0$.
\end{proof}

\section{Ternary Steane Code and its circuit}
In this section, we carry over the 7 qubit QECC by Steane \cite{PhysRevLett.77.793} to the ternary system and show its circuit implementation. The stabilizers for Steane code are
\begin{eqnarray*}
S_1 &=& I \otimes I \otimes I \otimes X \otimes X \otimes X \otimes X\\
S_2 &=& I \otimes X \otimes X \otimes I \otimes I \otimes X \otimes X\\
S_3 &=& X \otimes I \otimes X \otimes I \otimes X \otimes I \otimes X\\
S_4 &=& I \otimes I \otimes I \otimes Z \otimes Z \otimes Z \otimes Z\\
S_5 &=& I \otimes Z \otimes Z \otimes I \otimes I \otimes Z \otimes Z\\
S_6 &=& Z \otimes I \otimes Z \otimes I \otimes Z \otimes I \otimes Z
\end{eqnarray*}

X and Z correspond to $X_1$ and $Z_1$ respectively in ternary quantum systems \cite{gottesman1998fault}. Each Z operator is realized by a single CNOT gate, and X = H CNOT H, where H is the Hadamard gate. The ternary equivalent of CNOT gate is the $C+T$ \cite{PhysRevA.97.052302} gate, where
$$C+T: \sum\limits_{x,y \in \{0,1,2\}}\ket{x,(x+y)\%3}\bra{x,y}$$

Chrestenson basis \cite{Hurst1985-HURSTI-2} is the equivalent of Hadamard basis in ternary quantum system. Two conjugate Chrestenson bases $b_1$ and $b_2$ are defined as
\begin{eqnarray*}
\ket{+_i} & = & \frac{1}{\sqrt{3}}(\ket{0}+\ket{1}+\ket{2})\\
\ket{-_i} & = & \frac{1}{\sqrt{3}}(\ket{0}+\omega^i\ket{1}+\omega^{2i}\ket{2})\\
\ket{|_i} & = & \frac{1}{\sqrt{3}}(\ket{0}+\omega^{2i}\ket{1}+\omega^i\ket{2})
\end{eqnarray*}

The Chrestenson gates $Ch_1$ and $Ch_2$ convert a qutrit from computational basis to $b_1$ and $b_2$ respectively.
\begin{center}
	\begin{tabular}{ c  c}
		$Ch_1 = \frac{1}{\sqrt{3}}\begin{pmatrix}
		1 & 1 & 1\\
		1 & \omega & \omega^2\\
		1 & \omega^2 & \omega
		\end{pmatrix}$
		&
		$Ch_2 = \frac{1}{\sqrt{3}}\begin{pmatrix}
		1 & 1 & 1\\
		1 & \omega^2 & \omega\\
		1 & \omega & \omega^2
		\end{pmatrix}$
	\end{tabular}
\end{center}

One can verify that $Ch_1 Ch_2 = I$ and
\begin{center}
    $Ch_1 X_1 Ch_2 = Z_1$ \hspace{0.8cm} $Ch_1 X_2 Ch_2 = Z_2.$
\end{center}

The circuit for correcting a single bit error in a qutrit using the ternary Steane code is shown in Fig.~\ref{fig:steane}. The circuit for correcting bit errors follows from the stabilizers $S_4, S_5$, and $S_6$, and each $Z$ operator is realized using a $C+T$ gate. In the figure, $q_0$ to $q_6$ are data qutrits and the remaining are ancilla qutrits for syndrome measurement.

\begin{figure}
    \centering
    \caption{Circuit to correct a single bit error with ternary Steane code}
    \includegraphics[scale=0.26]{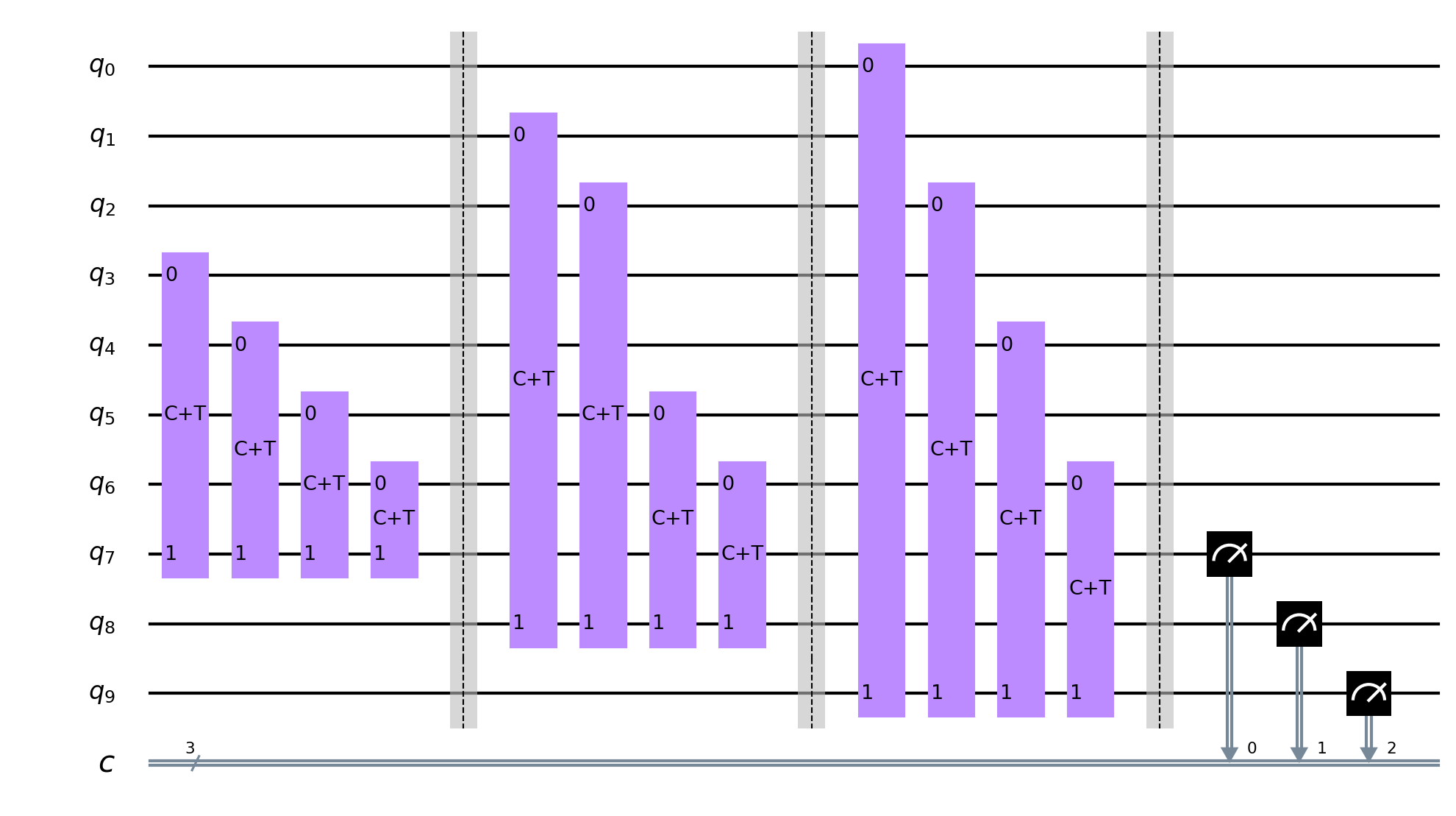}
    \label{fig:steane}
\end{figure}

To correct a single phase error, the qutrits must be converted to the Chrestenson basis. Therefore, the circuit for correcting a single phase error is similar to Fig.~\ref{fig:steane}, except that there will be a single $Ch_1$ gate at the beginning and a $Ch_2$ gate at the end of the circuit for each qutrit. Therefore, the total number of gates required to correct a single error using the ternary Steane code is 38 (24 $C+T$ gates and 14 Chrestenson gates). The depth of the circuit, defined as the maximum number of gates on any input to output path is 8, which is for $q_6$.

\section{7 qutrit degenerate QECC}
In our proposed 7 qutrit QECC, the information of a single qutrit $\ket{\psi}=\alpha\ket{0}+\beta\ket{1}+\gamma\ket{2}$ is encoded into seven qutrits as
$\ket{\psi}_L = \alpha\ket{0_L}+\beta\ket{1_L}+\gamma\ket{2_L}$, where
\begin{eqnarray*}
\ket{0_L} &=& \ket{0000000} + \ket{1020102} + \ket{2010201}\\
& + & \ket{0102010} + \ket{1122112} + \ket{2112211}\\
& + & \ket{0201020} + \ket{1221122} + \ket{2211221}\\
\ket{1_L} &=& \ket{1111111} + \ket{2101210} + \ket{0121012}\\
& + & \ket{1210121} + \ket{2200220} + \ket{0220022}\\
& + & \ket{1012101} + \ket{2002200} + \ket{0022002}\\
\ket{2_L} &=& \ket{2222222} + \ket{0212021} + \ket{1202120}\\
&+& \ket{2021202} + \ket{0011001} + \ket{1001100}\\
&+& \ket{2120212} + \ket{0110011} + \ket{1100110}
\end{eqnarray*}

This QECC is degenerate. For example, a single $Z_1$ error on the first and the fifth qutrits generate the same error state. We shall discuss the impact of degeneracy in error correction in the following subsection. This encoding scheme satisfies the necessary and sufficient conditions for error correction \cite{PhysRevLett.84.2525} which respectively states that
\begin{center}
    $\bra{0_L}\sigma \ket{0_L} = \bra{1_L}\sigma \ket{1_L} = \bra{2_L}\sigma \ket{2_L}$
\end{center}

for any error $\sigma \in \mathcal{E}$, and
\begin{center}
    $\bra{i_L}\sigma_m^{\dagger}\sigma_n\ket{j_L} = \delta_{ij}\alpha_{mn}$
\end{center}
for any errors $\sigma_m, \sigma_n \in \mathcal{E}$, where $i,j \in \{0,1,2\}$, $\alpha_{mn} \in \mathbb{C}$ and $\delta_{ij}$ is the Dirac-delta function.

\subsection{Correction of phase errors}
The stabilizers for correcting a single phase error are:
\begin{eqnarray*}
S_1 &=& X_1 \otimes I \otimes X_2 \otimes I \otimes X_1 \otimes I \otimes X_2\\
S_2 &=& I \otimes X_1 \otimes I \otimes X_2 \otimes I \otimes X_1 \otimes I
\end{eqnarray*}

The stabilizer $S_1$ and $S_2$ operate on two disjoint sets of qutrits. In other words, these two stabilizers partition the qutrits of the codeword in two sets, $g_1 = \{q_0,q_2,q_4,q_6\}$ and $g_2 = \{q_1,q_3,q_5\}$. The partitioning of the qutrits into two sets, and the action of $S_1$ and $S_2$ for phase errors occurring on qutrits of each set, are depicted in Table~\ref{tab:phase}.

\begin{table}[htb]
    \centering
    \caption{Partition of the qutrits into probable error subsets}
    \begin{tabular}{|c|c|c|c|c|}
    \hline
        & Type of error & $S_1$ & $S_2$ & Probable error qutrits\\
        \hline
        1 & \multirow{4}{*}{$Z_1$} & $\omega^2$ & 1 & $q_0$, $q_4$\\
        \cline{1-1} \cline{3-5}
        2 & & $\omega$ & 1 & $q_2$, $q_6$ \\
        \cline{1-1} \cline{3-5}
        3 & & 1 & $\omega^2$ & $q_1$, $q_5$ \\
        \cline{1-1} \cline{3-5}
        4 & & 1 & $\omega$ & $q_3$\\
        \hline
        5 & \multirow{4}{*}{$Z_2$} & $\omega$ & 1 & $q_0$, $q_4$ \\
        \cline{1-1} \cline{3-5}
        6 & & $\omega^2$ & 1 & $q_2$, $q_6$ \\
        \cline{1-1} \cline{3-5}
        7 & & 1 & $\omega$ & $q_1$, $q_5$ \\
        \cline{1-1} \cline{3-5}
        8 & & 1 & $\omega^2$ & $q_3$ \\
        \hline
    \end{tabular}
    \label{tab:phase}
\end{table}

It is easy to see that $\exists S \notin \{S_1,S_2\}$ such that $[S,S_i] = 0$, $i \in \{1,2\}$ and $wt(S) < 3$, which implies that the distance of this code is $< 3$. For example, the operator $Z_1 \otimes I \otimes Z_1 \otimes I \otimes I \otimes I \otimes I$ commutes with both the stabilizers, which implies that there exist phase errors on $q_0$ and $q_2$ which this code fails to distinguish. However, we show that it is possible to correct the error even without distinguishing them uniquely in certain cases.

Consider the error state $\ket{\Bar{\psi}}$ such that $S_1\ket{\Bar{\psi}} = \omega^2\ket{\Bar{\psi}}$. From Table~\ref{tab:phase}, we note that one cannot distinguish uniquely whether $Z_1$ error occurred on $q_0$, or $Z_2$ error occurred on $q_2$. However,
\begin{eqnarray*}
Z_1^0\ket{0_L} &=& \ket{0000000} + \omega\ket{1020102} + \omega^2\ket{2010201}\\
& + & \ket{0102010} + \omega\ket{1122112} + \omega^2\ket{2112211}\\
& + & \ket{0201020} + \omega\ket{1221122} + \omega^2\ket{2211221}\\
&=& Z_2^2\ket{0_L}
\end{eqnarray*}

where, $Z_i^j$ implies the error $Z_i$ acting on the $j^{th}$ physical qutrit, $i \in \{1,2\}$, $j \in \{0,1,\hdots,6\}$. $\ket{1_L}$ and $\ket{2_L}$ also show similar behavior (which is obvious since this QECC satisfies the necessary condition for error correction). Therefore, it is not necessary to distinguish between the two errors $Z_1^0$, and $Z_2^2$. Rather, when the stabilizer $S_1$ gives $\omega^2$ eigenvalue, correcting any one of the two errors is sufficient to correct the error on the codeword. Similar other scenarios are observable in Table~\ref{tab:phase}. Thus, although this QECC cannot distinguish between certain phase errors, it can still correct them perfectly.

We say that a stabilizer $S$ triggers for some error $E$, if for that error $E$ the eigenvalue of $S$ is not unity.

\begin{lemma}
\label{eq:group}
The proposed QECC can correct two simultaneous phase errors on two distinct qutrits $q_i$ and $q_j$, if these belong to two distinct sets defined above, i.e, $q_i \in g_1$ and $q_j \in g_2$.
\end{lemma}

\begin{proof}
The stabilizers $S_1$ and $S_2$ operate on two disjoint sets of qutrits, so a single phase error cannot trigger both of them. Both the stabilizers can trigger only when there are phase errors on two qutrits $q_i$ and $q_j$ such that $q_i \in g_1$ and $q_j \in g_2$. Each error can be individually detected and corrected according to Table~\ref{tab:phase}. Therefore, two simultaneous phase errors can be corrected if they occur on two qutrits belonging to disjoint sets $g_1$ and $g_2$.
\end{proof}

\begin{lemma}
\label{eq:group2}
The proposed QECC can correct upto $|g_i|$ simultaneous phase errors on the qutrits belonging to the same disjoint set $g_i$, $i \in \{1,2\}$.
\end{lemma}

\begin{proof}
Consider $0 \leq m \leq |g_i|$ phase errors occurring simultaneously on individual qutrits belonging to the set $g_i$, $i \in \{1,2\}$. The action of these phase errors on the codeword is
$$\bigotimes_{i=1}^m Z^j\ket{\psi} = \omega^{q}\ket{\psi}$$

where, $Z^j \in \{I,Z_1,Z_2\}$ is the phase error on the $j^{th}$ qutrit in $g_i$, and $q \in \{0,1,2\}$. Therefore. multiple phase errors on qutrits of $g_i$ behave like a single $Z_1$ or $Z_2$ error acting on a qutrit of $g_i$, which can be corrected as depicted in Table~\ref{tab:phase}.
\end{proof}

\begin{thm}
\label{eq:phase}
The proposed QECC can correct upto seven simultaneous phase errors on the codeword.
\end{thm}

\begin{proof}
The proof follows directly from Lemmata~\ref{eq:group} and ~\ref{eq:group2}.
\end{proof}


The circuit for correcting phase errors on the codeword, shown in Fig.~\ref{fig:phase}, has a depth of 4 along $q_2$, $q_3$ and $q_6$.

\begin{figure}[htb]
    \centering
    \caption{Circuit for correcting phase errors with the proposed QECC}
    \includegraphics[scale=0.26]{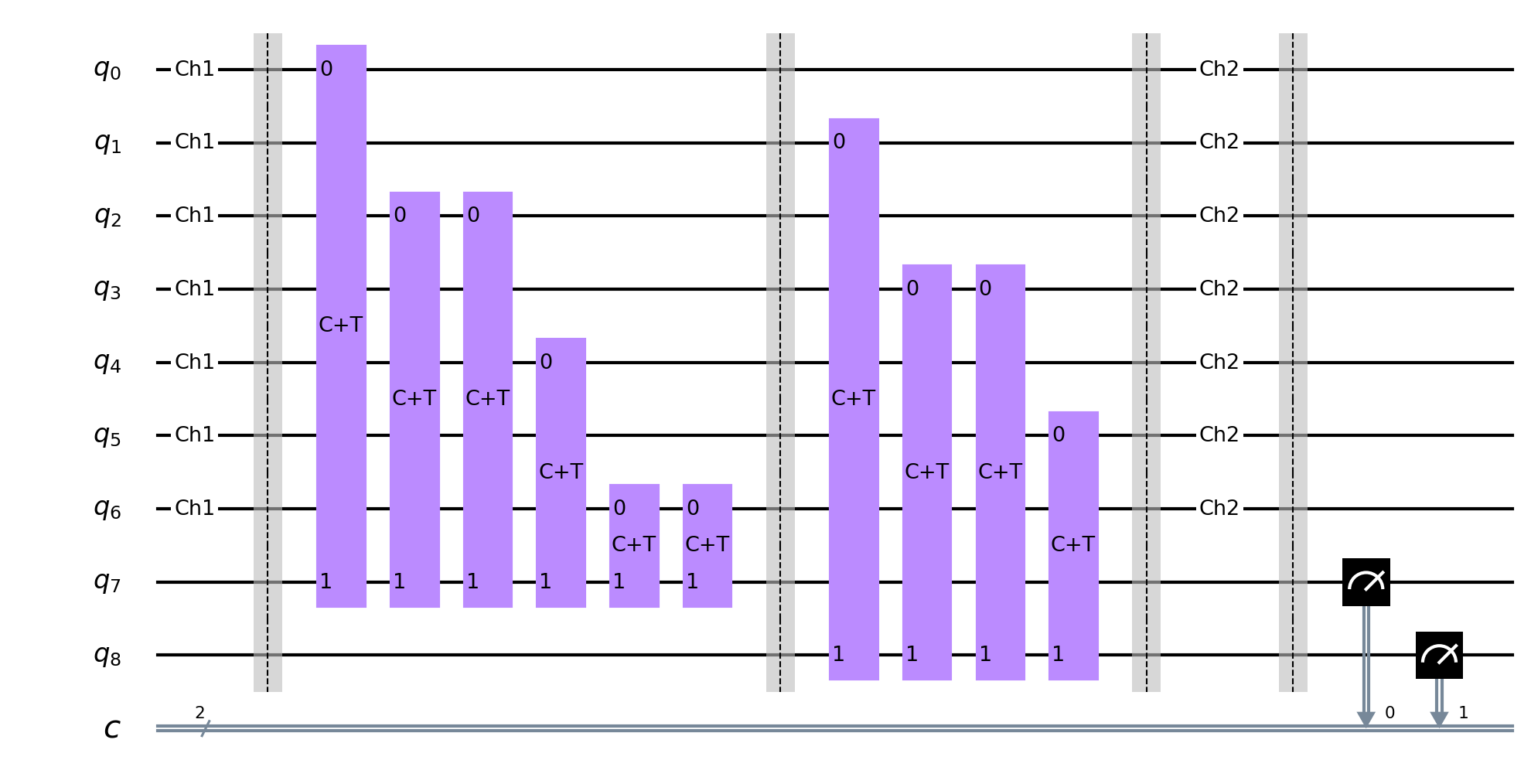}
    \label{fig:phase}
\end{figure}

\subsection{Correction of bit error}
The stabilizers for correcting bit errors are as follows:
\begin{eqnarray*}
S_3 &=& Z_1 \otimes Z_2 \otimes Z_1 \otimes Z_2 \otimes I \otimes I \otimes I\\
S_4 &=& I \otimes I \otimes I \otimes Z_1 \otimes Z_2 \otimes Z_1 \otimes Z_2\\
S_5 &=& I \otimes Z_1 \otimes Z_2 \otimes Z_1 \otimes Z_2 \otimes I \otimes I\\
S_6 &=& I \otimes I \otimes Z_1 \otimes Z_2 \otimes Z_1 \otimes Z_2 \otimes I
\end{eqnarray*}

Let us assume $X_i^j$ indicate error $X_i$ on the $j$-th qutrit, $i \in \{1,2\}$, $j \in \{0,1,...,6\}$. Table~\ref{tab:bit} shows the action of the stabilizers for a single $X_1$ error on different qutrits. Correction of $X_2$ errors will be similar. The non-unity values will be replaced by their complex conjugate (i.e. $\omega$ becomes $\omega^2$ and vice-versa).

\begin{table}[htb]
        \centering
        \caption{Correcting a single bit error CHECK}
        \begin{tabular}{|c|c|c|c|c|}
        \hline
        Error Type & $S_3$ & $S_4$ & $S_5$ & $S_6$\\
        \hline
        $X_1^0$ & $\omega$ & 1 & 1 & 1 \\
        \hline
        $X_1^1$ & $\omega$ & 1 & $\omega$ & 1\\
        \hline
        $X_1^2$ & $\omega$ & 1 & $\omega$ & $\omega$\\
        \hline
        $X_1^3$ & $\omega$ & $\omega$ & $\omega$ & $\omega$\\
        \hline
        $X_1^4$ & 1 & $\omega$ & $\omega$ & $\omega$\\
        \hline
        $X_1^5$ & 1 & $\omega$ & 1 & $\omega$\\
        \hline
        $X_1^6$ & 1 & $\omega$ & 1 & 1\\
        \hline
        \end{tabular}
        \label{tab:bit}
\end{table}

Note that the stabilizers $S_3, \hdots, S_6$ are not unique. For example, $S_3 = Z_2 \otimes Z_1 \otimes Z_2 \otimes Z_1 \otimes I \otimes I \otimes I$ is valid as well. The difference will be that some values of the Table~\ref{tab:bit} will be replaced by its complex conjugate.

The circuit for correcting a single bit error can be similarly obtained as for the phase error correction circuit ($Z_1$ and $Z_2$ realized using a single or two $C+T$ gates respectively). Four stabilizers for bit error correction leads to four ancilla qutrits. The depth of the circuit for bit error correction is 6 along $q_3$. Therefore, the depth of the entire error correction circuit is 10, which is along $q_3$.

\begin{lemma}
\label{eq:l2}
The set of stabilizers $S_3, \hdots, S_6$ can correct simultaneous bit errors occurring on $q_0$ and $q_6$.
\end{lemma}

\begin{proof}
From Table~\ref{tab:bit} it is evident that if both $q_0$ and $q_6$ have bit errors, then both the stabilizers $S_3$ and $S_4$ will trigger. These are the only two stabilizers to trigger in this scenario, and there are no other errors for which only $S_3$ and $S_4$ trigger. Hence, if these two stabilizers trigger, then it is possible to identify two simultaneous bit errors on $q_0$ and $q_6$.
\end{proof}

In the following, we show that it is possible to generate a set of stabilizers $S_3, \hdots, S_6$ such that they can correct a single bit error on any qutrit, as well as simultaneous bit errors on any pre-defined pair ($q_i,q_j$) of qutrits (except for three pairs). We first state the sufficient condition for correcting two simultaneous bit errors and discuss about the pairs for which simultaneous bit errors are not correctable.

\begin{lemma}
\label{eq:l3}
Two simultaneous bit errors on qutrits $q_i \neq q_j$ can be corrected if there exist two stabilizers $S_i \neq S_j$ which trigger for bit errors on $q_i$ and $q_j$ respectively, and there does not exist any single bit error, or simultaneous bit errors on any other pair of qutrits, for which only $S_i$ and $S_j$ trigger together.
\end{lemma}

\begin{proof}
Let $S_i$ and $S_j$ be two distinct stabilizers such that $S_i$ triggers for a bit error on $q_i$ and $S_j$ triggers for a bit error on $q_j$. Therefore, if bit errors occur on both $q_i$ and $q_j$ simultaneously, then the set of triggered stabilizers is $\{S_i,S_j\}$. Moreover, if there exist no other single error, or simultaneous bit errors on some different pair of qutrits for which the set of triggered stabilizers is $\{S_i,S_j\}$, then this set of triggered stabilizers can be uniquely decoded as simultaneous bit errors on $q_i$ and $q_j$.
\end{proof}

Lemma~\ref{eq:l3} is not a necessary condition for simultaneous bit error correction, because it may be possible to do so even when multiple stabilizers trigger for one or both the errors. However, it is not possible to correct simultaneous bit errors on every pair of qutrits, as elaborated in Lemma~\ref{eq:l4}.

\begin{lemma}
\label{eq:l4}
There are no stabilizers $S_i,~ S_j \in \{I,Z_1,Z_2\}^{\otimes 7}$ for the proposed QECC design such that it can correct two simultaneous bit errors if they occur on two qutrits of $g_2$.
\end{lemma}

\begin{proof}
Let two bit errors occur simultaneously on $q_i$ and $q_j$, where $q_i, q_j \in g_2$. Let the third qutrit in $g_2$ be $q_k$. Thus, we need at least 2 stabilizers $S_i$ and $S_j$, $S_i, S_j \in \{I, Z_1, Z_2\}^{\otimes 7}$ such that $S_i$ operates on $q_i$ and not on $q_j$, and vice versa for $S_j$. Furthermore, the eigenvalues of $S_i$ and $S_j$ together cannot be a syndrome for some other error. However,  if $S_i$ operates on $q_i$, then it must also operate on $q_k$ in order to commute with the stabilizer $S_2$ for phase error. Similarly $S_j$ must again operate on $q_k$ in order to commute with $S_2$. Any other stabilizer which operates on $q_k$, must also operate on $q_i$ or $q_j$ in order to commute with $S_2$. Therefore, $S_i, S_j$ together form the syndrome for $q_k$. If both $S_i$ and $S_j$ show non-identity eigenvalues, then it is not possible to distinguish whether both $q_i$ and $q_j$ are erroneous, or only $q_k$ is erroneous. Hence, two bit errors occurring simultaneously on two qutrits of $g_2$ cannot be corrected.
\end{proof}

We now provide an algorithm (Algorithm~\ref{alg}) to generate the set of stabilizers $S_3, \hdots, S_6$ in order to correct simultaneous bit errors in a pair of qutrits which is allowed by Lemma~\ref{eq:l4}. There are four stabilizers for correcting bit errors and seven qutrit positions for a 7-qutrit codeword. Let $S_i^j$ denote the operator acting on the $j^{th}$ qutrit for the $i^{th}$ stabilizer, $3 \leq i \leq 6$, $0 \leq j \leq 6$. Furthermore, let $p_k$ denote the $k^{th}$ qutrit of the codeword, $0 \leq k \leq 6$, and $d_i$ denotes the number of non-identity elements at the $i^{th}$ qutrit position for the stabilizers $S_3, \hdots, S_6$. For example, in the set of stabilizers shown at the beginning of Section IV B, $d_0 = 1$, $d_1 = 2$, $d_3 = 3$ and so on. Without loss of generality, in accordance to Lemma~\ref{eq:l3}, we shall consider $S_i = S_3$ and $S_j = S_4$ when correcting simultaneous bit errors on $q_i$ and $q_j$, $i < j$.

\begin{lemma}
\label{eq:alg}
Algorithm~\ref{alg} creates a valid set of stabilizers which can correct simultaneous bit errors on the pre-defined pairs of qutrits $q_i$ and $q_j$ excluding the ones by Lemma~\ref{eq:l4}.
\end{lemma}

\begin{proof}
We consider the proof in the following steps:
\begin{itemize}
    \item[(1)] For each stabilizer $S_i$, $3 \leq i \leq 6$, two operators placed in accordance with Lemma~\ref{eq:l1}, operate on qutrits of set $g_1$ ($g_2$). All these stabilizers $S_3, \hdots, S_6$ commute with $S_1$ and $S_2$. Further, it is verified that each $S_i$ stabilizes the codeword.
    \item[(2)] In Algorithm~\ref{alg}, $S_3$ ($S_4$) alone is chosen as a non-identity operator on qutrit $q_i$ ($q_j$), and only $S_3$ ($S_4$) can trigger for a bit error on $q_i$ ($q_j$). $S_3$ and $S_4$ together trigger when both $q_i$ and $q_j$ have simultaneous bit errors.
    \item[(3)] For unique detection of simultaneous bit errors, we require that for any qutrit pair ($q_k$,$q_l$), such that $q_k, q_l \in g_i$, $i \in \{1,2\}$, other than ($q_i$,$q_j$), either $d_k \neq d_l$, or if $d_k = d_l$, then the qutrits do not have non-identity operators for all the stabilizers at the same position. This condition is checked when putting the operators on the final stabilizer for both the sets $g_1$ and $g_2$. Therefore, even if such a scenario occurs upto stabilizer $S_5$, the tie will be broken in stabilizer $S_6$, leading to a unique pattern of non-unity eigenvalues for each error.
\end{itemize} 
\end{proof}

The time complexity of Algorithm~\ref{alg} is $\mathcal{O}(1)$, since the number of qutrits and the number of stabilizers is fixed.

\begin{algorithm}[H]
\caption{To generate the set of stabilizers $S_3, \hdots, S_6$ to correct simultaneous bit errors}
\label{alg}
\begin{algorithmic}[1]
\REQUIRE Qutrits $q_i$ and $q_j$ on which simultaneous bit errors are to be corrected.
\ENSURE The set of stabilizers $S_3, \hdots, S_6$.
\STATE $d_h \leftarrow 0$, $0 \leq h \leq 6$.
\STATE $S_3^i, S_4^j \leftarrow Z_1$.
\STATE $d_i \leftarrow d_i + 1$, $d_j \leftarrow d_j + 1$.
\STATE $S_q^i \leftarrow I$, $4 \leq q \leq 6$.
\STATE $S_q^j \leftarrow I$, $3 \leq q \leq 6$, $q \neq 4$.
\FORALL{stabilizers $S_3, \hdots, S_6$}
\FOR{$i \in \{1,2\}$}
\FORALL{qutrits in set $g_i$}
\STATE $q_k, q_l \leftarrow$ qutrits in $g_1$ for which $d_k$ and $d_l$ have the minimum and the second minimum values, and operators have not been assigned for these qutrits on the stabilizer. Break ties such that the operators at positions $P_k$ and $P_l$ of stabilizers $S_1$ and $S_2$ respectively are not equal. If such a scenario is not available, break ties arbitrarily. If one position $q_k$ is already occupied by a non-identity operator, choose only $q_l$ accordingly.
\IF{stabilizer is $S_6$}
\IF{$d_k \neq d_l$}
\STATE continue
\ELSE
\WHILE{both $q_k$ and $q_l$ has non-identity operators at the same positions for all prior stabilizers}
\STATE Discard $q_l$. New $q_l \leftarrow$ qutrit having the minimum value of $d_l$ apart from $q_k$ or the previous $q_l$. 
\ENDWHILE
\ENDIF
\ENDIF
\STATE Operator on $k^{th}$ qutrit $\leftarrow Z_u$, and operator on $l^{th}$ qutrit $\leftarrow Z_r$, $u, r \in \{1,2\}$ such that Lemma~\ref{eq:l1} is satisfied, and $S_i^k \otimes S_i^l \ket{\psi} = \ket{\psi}$.
\STATE $d_k \leftarrow d_k + 1$, $d_l \leftarrow d_l + 1$.
\ENDFOR
\ENDFOR
\ENDFOR
\end{algorithmic}
\end{algorithm}

We show an example of the stabilizer construction using Algorithm~\ref{alg} in Appendix.

\begin{thm}
\label{eq:bit}
A set of stabilizers $\in \{I,Z_1,Z_2\}^{\otimes 7}$ can be generated in $\mathcal{O}(1)$ time such that two simultaneous bit errors on a given pair $(q_i,q_j)$ of qutrits where both $q_i$ and $q_j \notin g_2$, can be corrected.
\end{thm}

\begin{proof}
The proof follows directly from Algorithm~\ref{alg} and Lemmata~\ref{eq:l3},~\ref{eq:l4} and ~\ref{eq:alg}.
\end{proof}

\section{Comparison of circuit depth and some remarks}
In Table~\ref{tab:comp} we compare the gate cost and the depth of the circuits of the 9-qutrit QECC \cite{PhysRevA.97.052302}, 6-qutrit approximate QECC (AQECC) \cite{majumdar2020approximate}, ternary Steane code, and our proposed QECC. The gate cost is in terms of $C+T$ and Chrestenson gates.

\begin{table}[htb]
    \centering
    \caption{Comparison of gate cost and depth of circuit}
    \begin{tabular}{|c|c|c|c|c|}
    \hline
    \multirow{2}{*}{} & 9-qutrit & 6-qutrit & Ternary & Proposed \\
    & QECC \cite{PhysRevA.97.052302} & AQECC \cite{majumdar2020approximate} & Steane & QECC\\
    \hline
      \# qutrits & 9 & 6 & 7 & 7\\
      \hline
      Gate cost for & \multirow{3}{*}{52} & \multirow{3}{*}{18} & \multirow{3}{*}{12} & \multirow{3}{*}{24}\\
      bit error & & & & \\
      correction & & & & \\
      \hline
      Gate cost for & \multirow{3}{*}{210} & \multirow{3}{*}{20} & \multirow{3}{*}{26} & \multirow{3}{*}{24}\\
      phase error & & & & \\
      correction & & & & \\
      \hline
      Total gate cost & 262 & 38 & 38 & 48 \\
      \hline
      Depth of circuit & 26 & 8 & 8 & 10 \\
      \hline
    \end{tabular}
    \label{tab:comp}
\end{table}

Both (i) the reduction in speed of computation due to error correction, and (ii) the decay in fidelity depend on the depth of the circuit \cite{arute2019quantum}. Although the circuit cost of our QECC is 10 more than that of the ternary Steane code, the increase in the depth of the circuit is only 2. Therefore, our QECC is not expected to have any significant performance degradation in terms of (i) and (ii). Moreover, the ternary Steane code can correct at most a single bit and a single phase error. The 6-qutrit AQECC can correct multiple phase errors, but it cannot correct a single bit error in all possible scenarios. Therefore our QECC surpasses both the 6-qutrit AQECC and the ternary Steane code in its ability to correct errors, without a significant increase in the depth of the circuit.

\begin{thm}
\label{eq:nqutrit}
It is possible to correct upto $n$ simultaneous phase errors for an $n$-qutrit QECC, $n \geq 3$, having stabilizers
\begin{equation*}
    S_1 = \bigotimes_{\substack{i = 0\\i~ is~ even}}^{n} X_i; \hspace{0.3cm} S_2 = \bigotimes_{\substack{j = 0\\j~ is~ odd}}^{n} X_j
\end{equation*}
\end{thm}

\begin{proof}
The stabilizers $S_1$ and $S_2$ partition the qutrits into two disjoint sets $g_1$ and $g_2$. The proof of Lemma~\ref{eq:group2} directly extends to any value of $|g_i|$, $i \in \{1,2\}$. Lemma~\ref{eq:group2}, together with Lemma~\ref{eq:group}, proves this theorem.
\end{proof}

Consider an $n$-qutrit QECC which can correct upto $t$ errors. The number of stabilizers in an $n$-qutrit QECC is $n-1$ \cite{gottesman1997stabilizer}. Theorem~\ref{eq:nqutrit} shows that for any $n \geq 3$, only two stabilizers are sufficient to correct upto $n$ simultaneous phase errors on the codeword. This opens up a rich field of using $n-3$ stabilizers to correct something more than just the $t$ bit errors. For our QECC, where $n = 7$ and $t = 1$, we could correct simultaneous bit errors on pre-defined qutrit pairs.

\begin{thm}
\label{eq:nobinary}
It is not possible to construct a linear degenerate CSS QECC for a qubit system that partitions the qubits into disjoint qubit sets.
\end{thm}

\begin{proof}
For an $n$-qubit linear CSS QECC, the codeword $\ket{0}^{\otimes n}$ must be present in the superposition state of the logical qubit. Any QECC which partitions the $n$ qubit codeword into disjoint sets $g_i$ must have the stabilizers $S_i$, $1 \leq i \leq n-1$ operating on disjoint qubits. Moreover, for a CSS code, the stabilizers can be partitioned into two sets, $\mathcal{S}_1$ and $\mathcal{S}_2$, such that stabilizers in $\mathcal{S}_1 \in \{I,\sigma_x\}^{\otimes n}$ and stabilizers in $\mathcal{S}_2 \in \{I,\sigma_z\}^{\otimes n}$. Therefore, the codewords are generated by operating the stabilizers in $\mathcal{S}_1$ multiple times on $\ket{0}^{\otimes n}$.

Let $q_{S_i}$ be the set of codewords which are generated by applying the stabilizer $S_i$ alone on $\ket{0}^{\otimes n}$. Also $\ket{0}^{\otimes n} \in \ket{0_L}$. Then $(\Pi_{i} S_i) \ket{0}^{\otimes n} = \ket{1}^{\otimes n}$, such that $S_i \in \mathcal{S}_1$. This implies that if such a set of stabilizers exist, then both $\ket{0}^{\otimes n}$ and $\ket{1}^{\otimes n} \in \ket{0_L}$. Hence, a QECC, in which the set of stabilizers operate on disjoint set of qubits, cannot exist for qubit systems.
\end{proof}

\section{Conclusion}
In this paper, we have proposed a 7-qutrit degenerate CSS QECC that can correct a single bit error and upto seven simultaneous phase errors. We have also shown that the stabilizers can be generated to correct upto two simultaneous bit errors on a pre-defined pair of qutrits (except for three pairs). Our QECC is optimal in the number of qutrits required to correct a single error with the CSS structure. Moreover, the depth of the circuit of our QECC is only two more than that of the ternary Steane code. We have also shown that this formulation can be extended to design any $n \geq 3$-qutrit QECC, but it is not possible to use this formulation technique to generate QECCs for qubit systems that can correct simultaneous phase errors. Therefore, our proposed code readily shows that there exist design techniques of QECC for the ternary quantum system which are more efficient than the ones which are a simple extension of binary quantum codes.

\section*{Acknowledgement}
All the figures have been generated using IBM qiskit simulator \cite{Qiskit}.
 
\bibliographystyle{unsrt}
\bibliography{main}

\appendix*
\section{Generating stabilizers using Algorithm~\ref{alg}}
Here we show a worked out example of generating the stabilizers for a pre-defined qutrit pair $(q_i,q_j)$ using Algorithm~\ref{alg}. For our example, we assume that the required qutrit pair is $(q_1,q_4)$.

\begin{enumerate}
    \item Initially $d_i = 0$, $\forall$ $i$. For brevity, we represent $d$ as a 7-tuple, all of whose entries are initialized to 0.
    \begin{center}
        $d = \begin{pmatrix}
        0 & 0 & 0 & 0 & 0 & 0 & 0
        \end{pmatrix}$
    \end{center}
    
    \item We shall start with putting $Z_1$ in positions 1 and 4 of stabilizers $S_3$ and $S_4$ respectively.
    \begin{eqnarray*}
    S_3 &=& - \quad Z_1 \quad - \quad - \quad - \quad - \quad -\\
    S_4 &=& - \quad - \quad - \quad - \quad Z_1 \quad - \quad -\\
    S_5 &=& - \quad - \quad - \quad - \quad - \quad - \quad -\\
    S_6 &=& - \quad - \quad - \quad - \quad - \quad - \quad -
    \end{eqnarray*}
    
    \item The tuple $d$ is updated as
    \begin{center}
        $d = \begin{pmatrix}
        0 & 1 & 0 & 0 & 1 & 0 & 0
        \end{pmatrix}$
    \end{center}
    
    \item Our requirement is that $S_3$ and $S_4$ alone will trigger for bit errors on $q_1$ and $q_4$ respectively. Therefore, for every other stabilizers, we put an identity in positions 1 and 4.
    \begin{eqnarray*}
    S_3 &=& - \quad Z_1 \quad - \quad - \quad I \quad - \quad -\\
    S_4 &=& - \quad I \quad - \quad - \quad Z_1 \quad - \quad -\\
    S_5 &=& - \quad I \quad - \quad - \quad I \quad - \quad -\\
    S_6 &=& - \quad I \quad - \quad - \quad I \quad - \quad -
    \end{eqnarray*}
    This step does not change $d$.
    
    \item The two disjoint sets are $g_1 = \{q_0, q_2, q_4, q_6\}$ and $g_2 = \{q_1, q_3, q_5\}$. We shall first fill up the stabilizer positions for qutrits in $g_1$. We need to choose two qutrit positions with minimum $d$ values. Let us select $q_0$ and $q_2$. From Lemma~\ref{eq:l1}, we can either use $Z_1$ on both, or $Z_2$ on both. However, one can verify that if all the four non-identity operators in a stabilizer are $Z_1$, then such a stabilizer will not stabilize the logical qubit. Take $\ket{1}_L$. A stabilizer which has four or three $Z_1$ operators, will produce a non-identity phase for the codeword $\ket{1111111}$ (and for others as well). Therefore, it is mandatory to have 2 $Z_1$ and 2 $Z_2$ operators in each stabilizer. Since we have already put one $Z_1$ operator, the operators in positions 0 and 2 must be $Z_2$.
    \begin{eqnarray*}
    S_3 &=& Z_2 \quad Z_1 \quad Z_2 \quad - \quad I \quad - \quad -\\
    S_4 &=& - \quad I \quad - \quad - \quad Z_1 \quad - \quad -\\
    S_5 &=& - \quad I \quad - \quad - \quad I \quad - \quad -\\
    S_6 &=& - \quad I \quad - \quad - \quad I \quad - \quad -
    \end{eqnarray*}
    The new tuple is $d = \begin{pmatrix}
        1 & 1 & 1 & 0 & 1 & 0 & 0
    \end{pmatrix}$
    
    \item Now if we look into $S_4$ for $g_1$, then here one position in $g_1$ is already filled up. So we shall consider some other position only. From the $d$-tuple, we select position 6. From Lemma~\ref{eq:l1}, we note that the operators in $q_4$ and $q_6$ must be the same in order for it to commute with stabilizer $S_1$. Therefore, the operator at position 6 must also be $Z_1$ (which implies that when we look into the operators for $g_2$ in $S_4$, both the operators must be $Z_2$).
    \begin{eqnarray*}
    S_3 &=& Z_2 \quad Z_1 \quad Z_2 \quad - \quad I \quad - \quad -\\
    S_4 &=& - \quad I \quad - \quad - \quad Z_1 \quad - \quad Z_1\\
    S_5 &=& - \quad I \quad - \quad - \quad I \quad - \quad -\\
    S_6 &=& - \quad I \quad - \quad - \quad I \quad - \quad -
    \end{eqnarray*}
    The new tuple is $d = \begin{pmatrix}
        1 & 1 & 1 & 0 & 1 & 0 & 1
    \end{pmatrix}$
    
    \item If we now look into the set $g_2$ for $S_3$, position 1 is already occupied with a non-identity operator. Therefore, we need to select some other position from $g_2$, say position 3. From Lemma~\ref{eq:l1}, and also from the argument discussed in step 5, the operator at position 3 of $S_3$ should be $Z_1$.
    \begin{eqnarray*}
    S_3 &=& Z_2 \quad Z_1 \quad Z_2 \quad Z_1 \quad I \quad - \quad -\\
    S_4 &=& - \quad I \quad - \quad - \quad Z_1 \quad - \quad Z_1\\
    S_5 &=& - \quad I \quad - \quad - \quad I \quad - \quad -\\
    S_6 &=& - \quad I \quad - \quad - \quad I \quad - \quad -
    \end{eqnarray*}
    The new tuple is $d = \begin{pmatrix}
        1 & 1 & 1 & 1 & 1 & 0 & 1
    \end{pmatrix}$
    
    \item Working out in the above-mentioned way for each stabilizer, we have the final set of stabilizers
    \begin{eqnarray*}
    S_3 &=& Z_2 \quad Z_1 \quad Z_2 \quad Z_1 \quad I \quad I \quad I\\
    S_4 &=& I \quad I \quad I \quad Z_2 \quad Z_1 \quad Z_2 \quad Z_1\\
    S_5 &=& Z_1 \quad I \quad Z_1 \quad Z_2 \quad I \quad Z_2 \quad I\\
    S_6 &=& Z_1 \quad I \quad I \quad Z_2 \quad I \quad Z_2 \quad Z_1
    \end{eqnarray*}
    The new tuple is $d = \begin{pmatrix}
        3 & 1 & 2 & 4 & 1 & 3 & 2
    \end{pmatrix}$
    
\end{enumerate}

If a bit error occurs on $q_1$, then only $S_3$ has a non-identity operator at position 3 and hence only $S_3$ will trigger. Similarly, for bit error on $q_4$ only $S_4$ will trigger. Finally, it can be verified that there exist no other single qutrit or two qutrit bit error for which only $S_3$ and $S_4$ trigger together. Therefore, this set of stabilizers can correct simultaneous bit errors on $q_1$ and $q_4$.

\end{document}